\documentclass[11pt]{article}\textwidth 6.5in\textheight 9in
\usepackage{amssymb}\usepackage[colorlinks]{hyperref}\usepackage{color}
	\usepackage{stmaryrd}\usepackage{mathrsfs}
	\usepackage{graphicx}
	\usepackage{amsmath}
 \usepackage{amsthm}
\usepackage{caption}

\begin{document}
\setlength{\captionmargin}{27pt}
\newcommand\hreff[1]{\href {http://#1} {\small http://#1}}
\newcommand\trm[1]{{\bf\em #1}} \newcommand\emm[1]{{\ensuremath{#1}}}

\setcounter{tocdepth}{3} 

\newtheorem{thr}{Theorem} 
\newtheorem{lmm}{Lemma}
\newtheorem{cor}{Corollary}
\newtheorem{con}{Conjecture} 
\newtheorem{prp}{Proposition}

\newtheorem{blk}{Block}
\newtheorem{dff}{Definition}
\newtheorem{asm}{Assumption}
\newtheorem{rmk}{Remark}
\newtheorem{clm}{Claim}
\newtheorem{example}{Example}

\newcommand{\ab}{a\!b}
\newcommand{\yx}{y\!x}
\newcommand{\yux}{y\!\underline{x}}

\newcommand\floor[1]{{\lfloor#1\rfloor}}\newcommand\ceil[1]{{\lceil#1\rceil}}

\newcommand{\lea}{<^+}
\newcommand{\gea}{>^+}
\newcommand{\eqa}{=^+}

\newcommand{\lel}{<^{\log}}
\newcommand{\gel}{>^{\log}}
\newcommand{\eql}{=^{\log}}

\newcommand{\lem}{\stackrel{\ast}{<}}
\newcommand{\gem}{\stackrel{\ast}{>}}
\newcommand{\eqm}{\stackrel{\ast}{=}}

\newcommand\edf{{\,\stackrel{\mbox{\tiny def}}=\,}}
\newcommand\edl{{\,\stackrel{\mbox{\tiny def}}\leq\,}}
\newcommand\then{\Rightarrow}

\newcommand\km{{\mathbf {km}}}\renewcommand\t{{\mathbf {t}}}
\newcommand\KM{{\mathbf {KM}}}\newcommand\m{{\mathbf {m}}}
\newcommand\md{{\mathbf {m}_{\mathbf{d}}}}\newcommand\mT{{\mathbf {m}_{\mathbf{T}}}}
\newcommand\K{{\mathbf K}} \newcommand\I{{\mathbf I}}

\newcommand\II{\hat{\mathbf I}}
\newcommand\Kd{{\mathbf{Kd}}} \newcommand\KT{{\mathbf{KT}}} 
\renewcommand\d{{\mathbf d}} 
\newcommand\D{{\mathbf D}}

\newcommand\w{{\mathbf w}}
\newcommand\Ks{\Lambda} \newcommand\q{{\mathbf q}}
\newcommand\E{{\mathbf E}} \newcommand\St{{\mathbf S}}
\newcommand\M{{\mathbf M}}\newcommand\Q{{\mathbf Q}}
\newcommand\ch{{\mathcal H}} \renewcommand\l{\tau}
\newcommand\tb{{\mathbf t}} \renewcommand\L{{\mathbf L}}
\newcommand\bb{{\mathbf {bb}}}\newcommand\Km{{\mathbf {Km}}}
\renewcommand\q{{\mathbf q}}\newcommand\J{{\mathbf J}}
\newcommand\z{\mathbf{z}}

\newcommand\B{\mathbf{bb}}\newcommand\f{\mathbf{f}}
\newcommand\hd{\mathbf{0'}} \newcommand\T{{\mathbf T}}
\newcommand\R{\mathbb{R}}\renewcommand\Q{\mathbb{Q}}
\newcommand\N{\mathbb{N}}\newcommand\BT{\{0,1\}}
\newcommand\FS{\BT^*}\newcommand\IS{\BT^\infty}
\newcommand\FIS{\BT^{*\infty}}\newcommand\C{\mathcal{L}}
\renewcommand\S{\mathcal{C}}\newcommand\ST{\mathcal{S}}
\newcommand\UM{\nu_0}\newcommand\EN{\mathcal{W}}

\newcommand{\supp}{\mathrm{Supp}}

\newcommand\lenum{\lbrack\!\lbrack}
\newcommand\renum{\rbrack\!\rbrack}

\renewcommand\qed{\hfill\emm\square}

\title{\vspace*{-3pc} On the Existence of Anomalies, The Reals Case}

\author {Samuel Epstein\footnote{JP Theory Group. samepst@jptheorygroup.org}}

\maketitle
\begin{abstract}
The Independence Postulate (IP) is a finitary Church-Turing Thesis, saying mathematical sequences are independent from physical ones. Modelling observations  as infinite sequences of real numbers, IP implies the existence of anomalies. 
\end{abstract}
\section{Introduction}
An outlier is an observation that is set apart from a population. There are many reasons that such anomalies occur, including measurement error and human error. However recent results have shown that outliers are ingrained into the nature of algorithms and dynamics. In \cite{Epstein21}, anomalies were proven to occur in sampling algorithms.  In \cite{EpsteinDerandom22}, anomalies were proven to exist in the outputs of probabilistic algorithms. They were also proven to be emergent computable ergodic dynamics on the Cantor space. In \cite{EpsteinDynamics22} anomalies were shown to emergent in a more general (but still computable) class of dynamics. These results were extended into computable metric spaces in the paper \cite{EpsteinThermo23}, showing computable dynamics in such spaces produce outliers. Furthermore, oscillations in algorithmic thermodynamic entropy were proven.

But what about measurements of systems that are too complex to be considered algorithmic? One example is the global weather system. One can attest to the fact that there are many strange formations that occur! To show that anomalies occur, one can use the Independence Postulate \cite{Levin84,Levin13}. The Independence Postulate is a finitary Church-Turing thesis, postulating that certain finite and infinite sequences cannot be easily be found with a short ``physical address''. In \cite{EpsteinExistAn23}, the Independence Postulate was used to show that observations, a.k.a. infinite sequences of natural numbers, that do not have outliers have high physical addresses. In other words, observations with no outliers cannot be found in nature.

In this paper, we extend these results to observations modeled by infinite sequences of reals. This enables a more natural modelling of phenomena such as fluid dynamics, etc. This paper reproduces the proof of infinite sequences in \cite{Epstein21}, but without using left-total machines, which require a lengthy explaination.

\section{Conventions}

The function $\K(x|y)$ is the conditional prefix Kolmogorov complexity. The mutual information between two strings $x,y\in\FS$, is $\I(x:y)=\K(x)+\K(y)-\K(x,y)$. For probability $p$ over $\N$, randomness deficiency is $\d(a|p,b)=\floor{-\log p(a)}-\K(a|\langle p\rangle, b)$ and measures the extent of the refutation against the hypothesis $p$ given the result $a$ \cite{Gacs21}. $\d(a|p)=\d(a|p,\emptyset)$. The amount of information that the halting sequence $\ch\in\IS$ has about $a\in\FS$, conditional to $y\in\FS$ is $\I(a;\ch|y)=\K(a|y)-\K(a|y,\ch)$. $\I(a;\ch)=\I(a;\ch|\emptyset)$. We use ${\lea} f$ to denote ${<}f{+}O(1)$ and ${\lel} f$ to denote ${<}f {+} O(\log(f{+}1))$. For a mathematical statement $A$, let $[A]=1$ if $A$ is true and $[A]=0$, otherwise. The chain rule gives $\K(x,y)\eqa \K(x|y,\K(y))+\K(y)$.
The following definition is from \cite{Levin74} .
\begin{dff}[Information]For infinite sequences $\alpha,\beta\in\IS$, their mutual information is defined to be 
$\I(\alpha\,{:}\,\beta){=}$
 $\log\sum_{x,y\in\FS}2^{\I(x:y)-\K(x|\alpha)-\K(y|\beta)}$.
\end{dff}

\noindent The Independence Postulate (\textbf{IP}), \cite{Levin84,Levin13}, is an unprovable inequality on the information content shared between two sequences. \textbf{IP} is a finitary Church Turing Thesis, postulating that certain infinite and finite sequences cannot be found in nature, a.k.a. have high “physical addresses”.\\

\noindent \textbf{IP}\textit{: Let $\alpha$ be a sequence defined with an $n$-bit mathematical statement, and a sequence $\beta$ can be located in the physical world with a $k$-bit instruction set. Then $\I(\alpha:\beta)<k+n+c$ for some small absolute 
constant $c$.}

\begin{lmm}[\cite{EpsteinShort23}]
\label{lmm:discrete}
For probability $p$ over $\N$, $D{\subset}\N$, $|D|=2^s$, $s < \max_{a\in D}\d(a|p)+\I(D;\ch)+O(\log\I(D;\ch)+\log \K(p))$. 
\end{lmm}

\begin{lmm}[\cite{EpsteinDerandom22}]
\label{lmm:cons}For partial computable $f$, $\I(f(x):\ch)\lea \I(x;\ch)+\K(f)$.
\end{lmm}

\section{Sets with Low Randomness Deficiencies}
A continuous probability $P$ over $\IS$ is identified with a function $P:\FS\rightarrow \R_{\geq 0}$, where $P(\emptyset)=1$ and $P(x)=P(x0)+P(x1)$. Randomness deficiency can be extended to continous probability measures with the following definition. 

\begin{dff}
The randomness deficiency of $\alpha\in\IS$ with respect to computable continuous probability measure $P$ is $\D(\alpha|P)=\sup_n-\log P(\alpha[0..n])-\K(\alpha[0..n]|\langle P\rangle)$. The term $\langle P\rangle$ is a program to compute $P$.
\end{dff}

\begin{rmk}
Let $\Omega = \sum\{2^{-\|p\|}:U(p)\textrm{ halts}\}$ be Chaitin's Omega and $\Omega^t = \sum\{2^{-\|p\|}:U(p)\textrm{ halts in time $t$}\}$. For a string $x$, let $BB(x)=\min \{t:\Omega^t>0.x+2^{-\|x\|}\}$. Note that $BB(x)$ is undefined if $0.x+2^{-\|x\|}>\Omega$. For $n\in \N$, let $\bb(n) = \max\{BB(x): \|x\|\leq n\}$. $\bb^{-1}(m) = \arg\min_n \{\bb(n-1)<m\leq \bb(n)\}$. Let $bb(n)=\arg\max_x\{BB(x) :\|x\|\leq n\}$.
\end{rmk}
\begin{lmm}
\label{lmm:rec}
For $n=\bb^{-1}(m)$, $\K(bb(n)|m,n)=O(1)$.
\end{lmm}
\begin{proof}
Enumerate strings of length $n$, starting with $0^n$, and return the first string $y$ such that $BB(y)\geq m$. This string $y$ is equal to $bb(n)$, otherwise $BB(y^-)$ is defined and $BB(y^-)\geq BB(y)\geq m$. Thus $\bb(n-1)\geq m$, causing a contradiction.
\end{proof}
\begin{prp}$ $\\
\label{prp:bb}
\vspace*{-0.5cm}
    \begin{enumerate}
    \item $\K(bb(n))\gea n$.
    \item $\K(bb(n)|\mathcal{H})\lea \K(n)$.
    \end{enumerate}
\end{prp}
The following lemma, while lengthy, is a series of straightforward application of inequalities.
\begin{lmm}
\label{lmm:main}
For continuous probability $P$ over $\IS$, $Z\subset\IS$, $|Z|=2^s$, $s\lel\max_{\alpha\in Z}\D(\alpha|P) + \I(\langle Z\rangle:\ch)+O(\log \K(P))$.
\end{lmm}
\begin{proof}
We relativize the universal Turing machine to $s$, which can be done due to the precision of the theorem. Let $Z_n = \{\alpha[0..n]:\alpha\in Z\}$ and $m=\arg\min_m |Z_m|=|Z|$. Let $n=\bb^{-1}(m)$ and $k=\bb(n)$. Let $p$ be a probability over $\FS$, where $p(x) = [\|x\|=k]P(x)$ and $\langle p\rangle = \langle k,P\rangle$. Using $D=Z_k$, Lemma \ref{lmm:discrete} relativized to $k$ produces $x\in Z_k$, where 
\begin{align*}
s &\lel -\log P(x)-\K(x|k,P)+\I(Z_k;\ch|k)+O(\log \K(P,k|k)) \\
&\lel -\log P(x)-\K(x|P) + \K(Z_k|k)+\K(k)-\K(Z_k|k,\ch)+O(\log \K(P)).
\end{align*}
Since $\K(k)\lea n+\K(n)$, by the chain rule,
\begin{align*}
&\K(Z_k|k)+\K(k)\\
\lea& \K(Z_k|k,\K(k))+\K(\K(k)|k)+\K(k)\\
<& \K(Z_k,k)+O(\log n)\\
<& \K(Z_k)+O(\log n).
\end{align*}
So
\begin{align*}
s&\lel -\log P(x)-\K(x|P) + \K(Z_k)-\K(Z_k|k,\ch)+O(\log n+\log \K(P)).
\end{align*}

Since $\K(k|n,\ch)=O(1)$, $\K(Z_k|\ch)\lea \K(Z_k|k,\ch)+\K(n)$. So
\begin{align*}
s&\lel -\log P(x)-\K(x|P)+\I(Z_k;\ch)+O(\log n+\log \K(P)).
\end{align*}
By Lemma \ref{lmm:rec}, $\K(bb(n)|Z_k)\lea \K(n)$ so by Lemma by \ref{lmm:cons} and Proposition \ref{prp:bb},
\begin{align*}
n&\lel\I(bb(n);\ch) \lel \I(Z_k;\ch)+\K(n)\lel\I(Z_k;\ch).
\end{align*}
So
\begin{align*}
s&\lel -\log P(x)-\K(x|P)+\I(Z_k;\ch)+O(\log \K(P)).
\end{align*}
By the definition of mutual information $\I$ between infinite sequences 
$$\I(Z_k;\ch)\lea \I(Z:\ch)+\K(Z_k|Z)\lel \I(Z:\ch)+\K(k|Z).$$
Now $m$ is simple relative to $Z$ and by Lemma \ref{lmm:rec}, $bb(n)$ is simple relative to $m$ and $n$. Furthermore $k$ is simple relative to $bb(n)$. Therefore
$\K(Z_k|Z) \lea \K(n)$. So
\begin{align*}
s &\lel -\log P(x) -\K(x|P) + \I(Z:\ch)+\K(n)+O(\log \K(P))\\
s &\lel \max_{\alpha\in Z}\D(\alpha|P) + \I(Z:\ch))+O(\log \K(P)).
\end{align*}
\qed
\end{proof}
$ $\\
Through careful observation, the above lemma can even be tightened to the following corollary
\begin{cor}
\label{cor}
For continuous probability $P$ over $\IS$, $Z\subset\IS$, $|Z|=2^s$, $s<\max_{\alpha\in Z}\D(\alpha|P) + \I(\langle Z\rangle:\ch)+O(\log \I(\langle Z\rangle:\ch)+ \log \K(P))$.
\end{cor}

\section{Observations as Reals}

We model observations as infinite sequences of reals in the interval $[0,1]$, or equivalently infinite sequences $\gamma$ of infinite sequences $\gamma_i\in\IS$, where each $\gamma_i$ is unique. Of course, in the real world, infinite sequences of observations do not exist. But infinite sequences model processes that are potentially never ending. Let $\langle\gamma\rangle\in\IS$ be a standard encoding of $\gamma$. Let $\gamma(n)\subset\IS$be the first $2^n$ infinite sequences of $\gamma$. The following theorem uses the simple fact that $\I(f(\alpha):\ch)\lea \I(\alpha:\ch)+\K(f)$, for $\alpha\in\IS$.

\begin{thr}$ $\\
\label{thr}
\noindent For probability $P$ over $\IS$, $\gamma\in{\IS}^\N$, let $t_{\gamma,P} = \sup_n(n-\K(n)-\max_{\alpha\in\gamma(n)}\D(\alpha|P))$. Then $t_{\gamma,P}\lel \I(\langle \gamma\rangle:\ch)+O(\log\K(P))$.
\end{thr}
\begin{proof}
By Corollary \ref{cor} applied to $\gamma(n)$,
\begin{align*}
n&< \max_{\alpha\in\gamma(n)}\D(\alpha|P)+\I(\gamma(n):\ch)+O(\log\I(\gamma(n):\ch)+\log\K(P))\\
n-\max_{\alpha\in\gamma(n)}\D(\alpha|P)&\lel +\I(\gamma(n):\ch)+O(\log\K(P))\\
n-\max_{\alpha\in\gamma(n)}\D(\alpha|P)-\K(n)&\lel +\I(\langle \gamma\rangle:\ch)+O(\log\K(P))\\
t_{\gamma,P}&\lel +\I(\langle \gamma\rangle:\ch)+O(\log\K(P)).
\end{align*}
\qed
\end{proof}

Let $k$ be a physical address of $\gamma$. $\ch$ can be described by a small mathematical statement. By Theorem \ref{thr} and \textbf{IP}, there is a small constant $c$ where
\begin{align*}
t_{\tau,\gamma}&\lel \I(\langle\gamma\rangle:\ch)+O(\log \K(P))\lel k+ c +O(\log \K(P)).
\end{align*}
It's hard to find observations with small anomalies and impossible to find observations with no anomalies.
\section{Discussion}
One avenue for future research is the relationship of outliers with different areas of physics. In thermodynamics, oscillations of thermodynamic entropy have been shown to occur \cite{EpsteinThermo23}. One area of study is into presence of outliers in quantum information theory. Recently, a Quantum EL theorem has been proven \cite{EpsteinQEL23}. Can this theorem be extended (as the EL Theorem was extended to the Outlier Theorem) to a statement saying streams of quantum qubits will contain outlying states?

\end{document}